\documentclass[11pt]{article}
\usepackage[english]{babel}
\usepackage{amssymb,amsmath,amsthm}
\newtheorem{theorem}{Theorem}[section]
\newtheorem{lemma}[theorem]{Lemma}
\newtheorem{proposition}[theorem]{Proposition}
\newtheorem{corollary}[theorem]{Corollary}
\newtheorem{question}[theorem]{Question}
\newtheorem{example}[theorem]{Example}

{\theoremstyle{definition}\newtheorem{definition}[theorem]{Definition}}
{\theoremstyle{definition}\newtheorem{remark}[theorem]{Remark}}

\newtheorem*{thmmp}{Proposition MP}
\newtheorem*{thmis}{Proposition I}

\numberwithin{equation}{section}

\def\C{{\mathbb C}}

\def\N{{\mathbb N}}
\def\Z{{\mathbb Z}}
\def\R{{\mathbb R}}

\def\T{{\mathbb T}}

\def\epsilon{\varepsilon}
\def\phi{\varphi}
\def\leq{\leqslant}
\def\geq{\geqslant}

{2}

\title{On Pauli Pairs}

\author{Stanislav Shkarin}

\date{}

\begin{document}

\maketitle

\begin{abstract} The state of a system in classical mechanics can be uniquely reconstructed if we know the positions and the momenta of all its parts. In 1958 Pauli has conjectured that the same holds for quantum mechanical systems. The conjecture turned out to be wrong. In this paper we provide a new set of examples of Pauli pairs, being the pairs of quantum states indistinguishable by measuring the spatial location and momentum. In particular, we construct a new set of spatially localized Pauli pairs.
\end{abstract}
\small 

\normalsize

\section{Introduction} 

In classical mechanics, one can determine the state of a system by measuring the position and the momentum of all its components. But what about the quantum mechanical systems? 

In his famous 1958 article \cite{pau} Paili asked a question whether the wave function of a system of spinless particles is determined up to a constant phase factor by the position and momentum probability densities. He conjectured that the answer to this question is affirmative. A number of authors, see, for instance, \cite{mo,mope,ism} have shown that this conjecture fails. The main purpose of this paper is to investigate to which extend it does. 

Recall that a state of a quantum mechanical system of spinless particles is given by a wave function $f$, which is a norm 1 vector in the complex Hilbert space $L^2(\R^n)$ determined uniquely up to a constant phase factor (=up to multiplying by $z\in\C$ with $|z|=1$). The density of the probability distribution of the positions of the particles in the system is given by $|f|^2$, while the probability distribution of their momenta has the density $(2\pi)^{-n}|\widehat f|^2$, where $\widehat f$ is the Fourier transform of $f$:
$$
\widehat f(x)=\int_{\R^n}e^{-i\langle x,y\rangle}f(y)\,dy,\ \ \ \text{where $\langle x,y\rangle=x_1y_1+{\dots}+x_ny_n$.}
$$
Measurements of positions and momenta can provide us only with the information about these distributions. Thus the two states $f$ and $g$ satisfying $|f|=|g|$ and $|\widehat f|=|\widehat g|$ (as functions) are physically indistinguishable. It would be natural to expect that such states must indeed be the same state. The above mentioned conjecture of Pauli says that it is possible to reconstruct $f\in L^2(\R^n)$ up to a scalar multiple knowing only $|f|$ and $|\widehat f|$. Rigorously speaking, his conjecture reads as follows.

\medskip

\noindent {\bf The Pauli Conjecture.} \ If $f,g\in L^2(\R^n)$ are such that $|f|=|g|$ and $|\widehat f|=|\widehat g|$ almost everywhere, then there is $z\in\C$ with $|z|=1$ such that $f=zg$.

\smallskip

Unfortunately, this conjecture fails. That is, there are pairs of functions (states) $f,g\in L^2(\R^n)$ such that $|f|=|g|$ and $|\widehat f|=|\widehat g|$ (almost everywhere) and $f$ and $g$ are linearly independent in $L^2(\R^n)$. We shall call them {\it Pauli pairs}. It is worth mentioning that the Pauli Conjecture turns out to be true in general position (in a certain sense) \cite{Jko}. The question of algorithmic reconstruction of $f$ given $|f|$ and $|\widehat f|$ was addressed in \cite{lokr}. Finally, we would like to mention the paper \cite{beism}, which introduces and studies an abstract operator theoretic generalization of the Pauli question.

It seems it was Moroz \cite{mo}, who first observed that certain symmetries provide Pauli pairs. For instance, if $f\in L^2(\R^n)$ is even, then the pair $(f,\overline{f})$ ($\overline{f}$ is the pointwise complex conjugate of $f$) satisfies $|f|=|\overline{f}|$ and $|\widehat{f}|=|\widehat{\overline{f}}|$. Thus $(f,\overline{f})$ is a Pauli pair provided $f$ is not a (complex) scalar multiple of a real-valued function. 

Note that the case $n=1$ is the most challenging one for producing Pauli pairs. For one, the isometry group of $\R^n$ becomes richer when $n$ grows and the symmetry based examples become easier to come by \cite{ism}. Another observation is that if $f_1,f_2\in L^2(\R^m)$ and $g_1,g_2\in L^2(\R^k)$ satisfy $|f_1|=|f_2|$, $|g_1|=|g_2|$, $|\widehat{f}_1|=|\widehat{f}_2|$ and $|\widehat{g}_1|=|\widehat{g}_2|$, then $|\psi_1|=|\psi_2|$ and $|\widehat{\psi}_1|=|\widehat{\psi}_2|$, where $\psi_1,\psi_2\in L^2(\R^{m+k})$ are defined by $\psi_1(t,x)=f_1(t)g_1(x)$ and $\psi_2(t,x)=f_2(t)g_2(x)$. Thus Pauli pairs for smaller $n$ generate Pauli pairs for bigger $n$.

Throughout this paper, for the sake of brevity, we shall drop the 'almost everywhere' refrain. When we write $f=g$ for $f$ and $g$ being Lebesgue measurable functions, we always mean $f=g$ almost everywhere. Similarly, we say that $f$ is non-constant if there is no constant to which $f$ equals almost everywhere, we say that $f:\R\to \C$  is $T$-periodic with $T>0$ if $f(x+T)=f(x)$ for almost all $x\in\R$ etc. Everywhere below, we shall use the fairly standard notation: \begin{align*}
&\text{$\T=\{z\in\C:|z|=1\}$, \ $\R_+=[0,\infty)$,}
\\
&\text{$\N=\{1,2,\dots\}$ is the set of positive integers and} 
\\
&\text{$\Z$ is the set of all integers.} 
\end{align*}
In this article, we produce new sets of Pauli pairs in the case $n=1$. We start by reminding the following observation of Moroz and Perelomov \cite{mope}.

\begin{thmmp} Let $\rho:\R\to\R_+$ and $\phi:\R\to\R$ be such that $\rho\in L^2(\R)$ and $\phi$ is Borel measurable. Assume also that there is $a\in\R$ such that $\rho(x)=\rho(a-x)$ on $\R$. Then
$|f_1|=|f_2|$ and $|\widehat{f}_1|=|\widehat{f}_2|$, where $f_1,f_2\in L^2(\R)$ are defined by the formulas $f_1(x)=\rho(x)e^{i\phi(x)}$ and $f_2(x)=\rho(x)e^{-i\phi(a-x)}$. In particular, if $e^{i(\phi(x)+\phi(a-x))}$ is non-constant, then $(f_1,f_2)$ is a Pauli pair.
\end{thmmp}

This proposition goes as far as one can get in generalizing the above mentioned even function idea in the case $n=1$.
Moroz and Perelomov \cite{mope} conjectured that all Pauli pairs in $L^2(\R)$ are given by Proposition~MP. Ismagilov \cite{ism} proved them wrong. Namely, he proved the following fact. The proof is nice and short, so we reproduce it in a slightly different form for the sake of the reader's convenience.

\begin{thmis} Let $T>0$, $\phi:\R\to\T$ be a $T$-periodic Borel measurable function and $g\in L^2(\R)$ be a non-zero function supported on an interval of the length $T/2\pi$. For each $a\in\R$, consider the function $f_a(x)=\widehat g(x)\phi(x-a)$. Then the functions $|f_a|$ and $|\widehat f_a|$ do not depend on the choice of $a$.
If additionally, $\phi$ is not of the form $\phi(x)=e^{iTk(x-d)/(2\pi)}$
with some $d\in\R$ and $k\in\Z$, then there is $c=c(\phi)>0$ such that $(f_a,f_b)$ is a Pauli pair whenever $0<|a-b|<c$.
\end{thmis}

\begin{proof} Without loss of generality, we may assume that $T=2\pi$. Since  $\phi$ is a $2\pi$-periodic function and belongs to $L^2[0,2\pi]$, we can write its Fourier series expansion 
$$
\phi(x)=\sum_{k=-\infty}^\infty c_ke^{ikx},
$$
where the Fourier coefficients $c_k\in\C$ satisfy $\sum\limits_{k=-\infty}^\infty|c_k|^2<\infty$. An easy computation yields
$$
f_a(x)=\sum_{k=-\infty}^\infty c_ke^{-ika}e^{ikx}\widehat g(x)\ \ \text{and}\ \
\widehat{f}_a(y)=2\pi\sum_{k=-\infty}^\infty c_ke^{-ika}g(-k-y).
$$
Since the supports of the functions $y\mapsto g(-k-y)$ do not intersect, we get
$$
|\widehat{f}_a(y)|=2\pi \sum_{k=-\infty}^\infty |c_k||g(-k-y)|.
$$
Thus $|\widehat{f}_a|$ does not depend on $a$. Obviously $|f_a|=|\widehat g|$ also does not depend on $a$.

Finally, it is an easy exercise to see that if the measurable function $\phi$ is not of the form $\phi(x)=e^{ik(x-d)}$ for some $d\in\R$ and $k\in\Z$, then it is impossible for $\phi$ to be a scalar multiple of each of $\phi(\cdot-c_n)$ with $\{c_n\}$ being a sequence of positive numbers converging to $0$. Hence there is $c>0$ such that the two functions $x\mapsto \phi(x-a)$ and $x\mapsto \phi(x-b)$ are linearly independent provided $0<|a-b|<c$. Since $\widehat g$ is non-zero and analytic on $\R$, $f_a$ and $f_b$ are linearly independent. Thus $(f_a,f_b)$ is a Pauli pair whenever $0<|a-b|<c$.
\end{proof}

Recall that the Schwartz space $S(\R)$ consists of $f\in C^\infty(\R)$ such that 
$$
\text{$p_{n,k}(f)=\sup\{(1+|x|)^k|f^{(n)}(x)|:x\in\R\}<\infty$ for every non-negative integers $n$ and $k$}. 
$$
The norms $p_{n,k}$ define a Fr\'echet space topology on $S(\R)$. It is well-known and easy to show that $S(\R)$ is a dense in $L^2(\R)$ linear subspace invariant for both the Fourier transform and its inverse.

\begin{remark}\label{rr1} It is easy to see that if under the conditions of Proposition~I, both $\phi$ and $g$ are infinitely differentiable, then each $f_a$ belongs to $S(\R)$ and therefore each $\widehat{f}_a$ belongs to $S(\R)$. If additionally $\phi$ is real-analytic, then so is each $f_a$. On the other hand $\widehat{f}_a$ is never real-analytic. Note also that Proposition~I provides a one-parametric family of functions in $L^2(\R)$ each pair of which is a Pauli pair.

Another point is that since $\widehat g$ is non-zero and analytic and $|\phi|=1$, the support of each $f_a$ in Proposition~I is unbounded (it is the entire real line actually). Next, since there are no trigonometric polynomials $p$ apart from $p(x)=e^{ik(x-a)}$ with $k\in\Z$ and $a\in\R$ satisfying $|p|\equiv1$, the support of each $\widehat{f}_a$ in Proposition~I is also unbounded. Indeed, the last observation forbids $\phi$ to be a trigonometric polynomial and therefore infinitely many of the Fourier coefficients $c_k$ of $\phi$ are non-zero. The explicit expression for $\widehat {f}_a$ in the proof of Proposition~I immediately entails the unboundedness of the support of $\widehat{f}_a$.
\end{remark}

\subsection{Main results} 

The first result of this paper shows that we can go much further than Proposition~I suggests. 

\begin{definition}\label{uzdset} We say that a set $S\subset L^2(\R)$ is an
{\it ultimate zero divisor set} (an {\it UZD-set} for short) if the cardinality of $S$ is at least 2, each $f\in S$ is a non-zero element of $L^2(\R)$ and $fg=\widehat f\widehat g=0$ for every distinct $f,g\in S$.
\end{definition}

\begin{theorem}\label{th1} There is a countable infinite UZD-set $S$ in $L^2(\R)$ such that $S\subset S(\R)$.
\end{theorem}

The above theorem is an enormous source of Pauli pairs.
Indeed, if $\{f_n\}_{n\in\N}$ is a sequence of pairwise distinct non-zero functions in $L^2(\R)$ such that $\{f_n:n\in\N\}$ is a UZD-set and $c=\{c_n\}_{n\in\N}$ is a sequence of positive numbers such that $\sum\limits_{n=1}^\infty c_n^2\|f_n\|^2<\infty$, then for every sequence $w=\{w_n\}_{n\in\N}$ in $\T$, we have
$$
|g_w|=\sum_{n=1}^\infty c_n|f_n|\ \ \text{and}\ \ |\widehat {g}_w|=\sum_{n=1}^\infty c_n|\widehat{f}_n|,\ \ \text{where}\ \ g_w=\sum_{n=1}^\infty w_nc_nf_n\in L^2(\R).
$$
Also, $g_w$ and $g_s$ are linearly dependent precisely when the sequences $w$ and $s$ of elements of $\T$ are proportional. Thus Theorem~\ref{th1} provides a family of pairwise distinct functions smoothly parametrized by infinitely many numbers from $\T$ such that each pair of the family is a Pauli pair. 

\begin{definition}\label{pstate} We say that $f\in L^2(\R^n)$ is a {\it Pauli state} if there is $g\in L^2(\R^n)$ such that $(f,g)$ is a Pauli pair.
\end{definition}

\begin{corollary}\label{psta} There is a closed infinite dimensional subspace $L$ of $L^2(\R)$ such that every non-zero vector in $L$ is a Pauli state. 
\end{corollary}

\begin{proof} By Theorem~\ref{th1}, there is a sequence $\{f_n\}_{n\in\N}$ in $L^2(\R)$ such that $\|f_n\|=1$ for each $n\in\N$ and $f_kf_m=\widehat{f}_k\widehat{f}_m=0$ whenever $k\neq m$. For each $n\in\N$, set $\phi_n=\frac1{\sqrt 2}(f_{2n-1}+f_{2n})$ and $\psi_n=\frac1{\sqrt 2}(f_{2n-1}-f_{2n})$. Then $\{\phi_n\}\cup \{\psi_n\}$ is an orthonormal sequence in $L^2(\R)$. Let $L$ be the closed linear span of $\{\phi_n:n\in\N\}$. Clearly, $L$ is a closed infinite dimensional subspace of $L^2(\R)$. It remains to verify that each non-zero $f\in L$ is a Pauli state. 

Let $f\in L\setminus\{0\}$. Since $\{\phi_n\}$ is an orthonormal basis in $L$, $f=\sum\limits_{n=1}^\infty c_n\phi_n$, where $c_n=\langle f,\phi_n\rangle$. Since $\{\psi_n\}$ is an orthonormal sequence in $L^2(\R)$ and $\sum\limits_{n=1}^\infty |c_n|^2<\infty$, the formula $g=\sum\limits_{n=1}^\infty c_n\psi_n$ defines a function in $L^2(\R)$. Since $\{\phi_n\}\cup \{\psi_n\}$ is orthonormal, $f$ and $g$ are orthogonal vectors in $L^2(\R)$ of the same positive norm and therefore $f$ and $g$ are linearly independent. Now the definitions of $f$ and $g$ together with the disjointness of the supports of $f_k$ and of $\widehat{f}_k$ yield
$$
|f|=|g|=\frac1{\sqrt 2}\sum_{n=1}^\infty |c_n|(|f_{2n-1}|+|f_{2n}|)\ \ \text{and}\ \ 
|\widehat f|=|\widehat g|=\frac1{\sqrt 2}\sum_{n=1}^\infty |c_n|(|\widehat{f}_{2n-1}|+|\widehat{f}_{2n}|).
$$
Hence $(f,g)$ is a Pauli pair and therefore $f$ is a Pauli state. 
\end{proof}

We were unable to prove the full fledged analog of Theorem~\ref{th1} for periodic functions. Still, the following result holds. As usual, we identify the functions $f\in L^2[0,2\pi]$ with $2\pi$-periodic functions on $\R$ whose restriction to $[0,2\pi]$ is square integrable. For $f\in L^2[0,2\pi]$, the Fourier coefficients are given by
$$
\widetilde{f}(n)=\frac1{2\pi}\int_{0}^{2\pi} f(t)\,e^{-int}\,dt,\ \ \ n\in\Z.
$$
We say that a set $S\subset L^2[0,2\pi]$ is an {\it UZD-set} if the cardinality of $S$ is at least 2, each $f\in S$ is a non-zero element of $L^2[0,2\pi]$ and $fg=0$ and $\widetilde{f}(n)\widetilde{g}(n)=0$ for all $n\in\Z$ for every distinct $f,g\in S$.

\begin{theorem}\label{th2} For every integer $n\geq2$, there is an $n$-element UZD-set $S\subset L^2[0,2\pi]$ consisting of $2\pi$-periodic infinitely differentiable functions.
\end{theorem}

The question whether there is an infinite UZD-set in $L^2[0,2\pi]$ remains open.

Note that the supports of $f$ and $\widehat f$ for any member $f$ of a UZD-set in $L^2(\R)$ are necessarily unbounded. This happens because the Fourier transform of a function (or a distribution for what it matters) with bounded support extends to an entire function on the complex plane. Hence the support of every member of a Pauli pair provided by Theorem~\ref{th1} is unbounded. By Remark~\ref{rr1}, the Pauli pairs provided by Proposition~I must also consist of functions with unbounded supports regardless whether we take the pair $(f_a,f_b)$ or the pair $(\widehat{f}_a,\widehat{f}_b)$. One might think that there are no Pauli pairs in $L^2(\R)$ of functions with bounded supports. This is however not the case. In the notation of Proposition~MP, the
Pauli pair $(f_1,f_2)$ may easily consist of functions with bounded supports. Moreover, the supports of $\widehat{f}_1$ and $\widehat{f}_2$ may also happen to be bounded (not simultaneously with the supports of $f_1$ and $f_2$, of course). Moroz has asked the author whether a bounded support Pauli pair $(f,g)$ in $L^2(\R)$ must only arise from Propositions~MP either as $(f_1,f_2)$ or as $(\widehat{f}_1,\widehat{f}_2)$. The reason to conjecture such a thing is that the boundedness of the support of $f\in L^2(\R)$ puts a serious constrain on the behavior of the Fourier transform $\widehat f$. Namely it forces $\widehat f$ to be an entire function of exponential type and it is not very common for two distinct entire functions to have the same absolute value on the real axis.
We answer the question of Moroz negatively by means of an example of a Pauli pair consisting of two step functions.
Before presenting the example, we would like to point out the constraints on a Pauli pair imposed by Proposition~MP.
First, observe that for a Pauli pair $(f_1,f_2)$ provided by Proposition~MP, the graph of $|f_1|=|f_2|$ has a vertical line of symmetry. Furthermore, there are $a,b\in\R$ such that $\widehat{f}_2(x)=e^{i(ax+b)}\overline{\widehat{f}_1(x)}$ on $\R$. This allows us to introduce the following concept.

\begin{definition}\label{mpp} Let $(f,g)$ be a Pauli pair of functions in $L^2(\R)$. We say that $(f,g)$ is an {\it MP$^1$-pair} if the graph of $|f|=|g|$ has a vertical line of symmetry (that is $|f(a-x)|\equiv|f(x)|$ for some $a\in\R$). We say that $(f,g)$ is an {\it MP$^2$-pair} if there are $a,b\in\R$ such that $g(x)=e^{i(ax+b)}\overline{f(x)}$ on $\R$.
\end{definition}

The point of the above definition is that due to the preceding remark, a Pauli pair $(f,g)$ can be obtained by means of using Proposition~MP only if it is either an MP$^1$-pair or an MP$^2$-pair. We shall answer the question of Moroz by providing a Pauli pair of (bounded support) step functions with 4 steps, which is neither an MP$^1$-pair nor an MP$^2$-pair.

For $c\in\C^n$, we define the step function $h_c\in L^2(\R)$ to be constant $0$ on $(-\infty,0)$ and on $[n,\infty)$ and constant $c_j$ on $[j-1,j)$ for $1\leq j\leq n$:
$$
h_c(x)=\left\{\begin{array}{ll}0&\text{if $x\notin [0,n);$}\\ c_j&\text{if $j-1\leq x<j$ \ $(1\leq j\leq n)$.}
\end{array}\right.
$$

\begin{remark}\label{rererere} Applying Definition~\ref{mpp} to a pair of step functions, we immediately observe the following. Assume that $b,c\in\C^n$ and that $(h_b,h_c)$ is a Pauli pair of step functions with $b_1c_1b_nc_n\neq0$. Then $(h_b,h_c)$ is an MP$^1$-pair if and only if
\begin{equation}\label{mp1}
|b_j|=|b_{n+1-j}|\ \ \text{for}\ \ 1\leq j\leq n.
\end{equation}
Furthermore, $(h_b,h_c)$ is an MP$^2$-pair if and only if
\begin{equation}\label{mp2}
c_j=w\overline{b_j}\ \ \text{for}\ \ 1\leq j\leq n\ \ \text{for some $w\in\T$ independent on $j$.}
\end{equation}
\end{remark}

The above remark provides easy means to confirm that a Pauli pair of step functions is not given by Proposition~MP.

\begin{example}\label{ex1} Let $f,g:\R\to\C$ be defined by the formula
$$
f(x)=\left\{\begin{array}{ll}0&\text{if $x\notin [0,4);$}\\ 1&\text{if $0\leq x<1;$}\\
\frac32 e^{4\pi i/3}&\text{if $1\leq x<2;$}\\ 3e^{4\pi i/3}&\text{if $2\leq x<3;$}\\ -2&\text{if $3\leq x<4$.}
\end{array}\right.\qquad\text{and}\qquad g(x)=\left\{\begin{array}{ll}0&\text{if $x\notin [0,4);$}\\ 1&\text{if $0\leq x<1;$}\\
-\frac32&\text{if $1\leq x<2;$}\\ 3&\text{if $2\leq x<3;$}\\ -2&\text{if $3\leq x<4$.}
\end{array}\right.
$$
Then $(f,g)$ is a Pauli pair.
\end{example}

Remark~\ref{rererere} easily ensures that the Pauli pair in the above example does not come from Propositions~MP. Thus Example~\ref{ex1} answers the mentioned question of Moroz negatively. It also provides a Pauli pair, which can not be obtained by any previously known construction.

We prove Theorems~\ref{th1} and~\ref{th2} in Section~2. We study the Pauli pairs of step functions and justify Example~\ref{ex1} in Section~3. Section~4 is devoted to concluding remarks and open questions.

\section{Proof of Theorem~\ref{th1}}

For functions $a,b\in L^2(\R)$ such that $b$ is supported on $[0,1]$ and $a$ is supported on $[0,2\pi]$ we consider the function
\begin{equation}
f_{a,b}\in L^2(\R),\qquad f_{a,b}(x)=\sum_{k=-\infty}^\infty \widehat a(k)b(x+k).
\label{fab}
\end{equation}
Since $\frac{\widehat{a}(k)}{2\pi}=\widetilde{a}(k)$ for $k\in\Z$, $\sum\limits_{k=-\infty}^\infty |\widehat{a}(k)|^2<\infty$, which in view of the disjointness of the supports of the functions $x\mapsto b(x+k)$ guarantees the convergence in $L^2(\R)$ of the series in the above display together with the identity $\|f_{a,b}\|=2\pi\|b\|\|a\|$. A direct computation yields
\begin{equation}
\widehat{f}_{a,b}(y)=\widehat{b}(y)\sum_{k=-\infty}^\infty \widehat a(k)e^{iky}=2\pi\widehat{b}(y) a_{\rm per}(y),
\label{fabf}
\end{equation}
where $a_{\rm per}$ is the $2\pi$-periodic function on $\R$, which coincides with $a$ on $[0,2\pi)$: $a_{\rm per}(y)=\sum\limits_{k=-\infty}^\infty a(y-2\pi k)$.

The formulas (\ref{fab}) and (\ref{fabf}) immediately imply that
\begin{equation}
a_1a_2=0\ \Longrightarrow\ \widehat{f}_{a_1,b_1}\widehat{f}_{a_2,b_2}=0\ \ \text{and}\ \
b_1b_2=0\ \Longrightarrow\ f_{a_1,b_1}f_{a_2,b_2}=0.
\label{zd}
\end{equation}

As usual, by the symbol $s(\Z)$, we denote the space of complex sequences $x=\{x_n\}_{n\in\Z}$ such that for every $k\in\N$, $p_k(x)=\sup\{|x_n|(1+|n|)^k:n\in\Z\}<\infty$. The space $s(\Z)$ is known as the space of rapidly decreasing (bilateral) sequences. It is well-known and easy to see that a $2\pi$-periodic integrable on $[0,2\pi]$ function is infinitely differentiable if and only if its sequence of Fourier coefficients belongs to $s(\Z)$. This fact together with (\ref{fab}) implies (via an easy estimate) that $f_{a,b}\in S(\R)$ if and only if both $a$ and $b$ are infinitely differentiable.

{\bf We are ready to prove Theorem~\ref{th1}.} \ Pick a sequence $\{I_n\}_{n\in\N}$ of pairwise disjoint closed subintervals of $(0,1)$ and a sequence $\{J_n\}_{n\in\N}$ of pairwise disjoint closed subintervals of $(0,2\pi)$. For each $n\in\N$ pick non-zero $a_n,b_n\in C^\infty(\R)$ such that the support of $a_n$ is contained in $J_n$ and the support of $b_n$ is contained in $I_n$. By the last observation, each $f_n=f_{a_n,b_n}$ is a non-zero element of $S(\R)$. Since the supports of $a_n$ are pairwise disjoint, $a_na_m=0$ whenever $n\neq m$. Similarly, $b_nb_m=0$ whenever $n\neq m$. By (\ref{zd}), $f_nf_m=0$ and $\widehat{f}_n\widehat{f}_m=0$ whenever $n\neq m$. Thus $\{f_n:n\in\N\}$ is a countable infinite UZD-set in $L^2(\R)$ consisting of elements of $S(\R)$. The proof of Theorem~\ref{th1} is complete.

\section{Proof of Theorem~\ref{th2}}

Recall that every $x\in[0,1)$, which is not a binary rational number has a unique infinite binary presentation:
$$
x=\sum_{n=1}^\infty \frac{x_n}{2^n},\ \ \text{where $x_n\in\{0,1\}$ are the binary digits of $x$.}
$$
Let $r_0$ be the constant $1$ function on $\R$. For $n\in\N$, we define the function $r_n:\R\to\R$ by
$$
r_n(x)=\left\{\begin{array}{ll}0&\text{if $x$ is binary rational;}\\ 1&\text{if $\{x\}_n=0$;}\\
-1&\text{if $\{x\}_n=1$,}
\end{array}\right.\qquad\text{where $\{x\}\in [0,1)$ is the fractional part of $x$.}
$$
The functions $r_n$ for $n\geq 0$ are known as the Rademacher functions. Obviously, each $r_n$ is $1$-periodic.

\begin{lemma}\label{rade} Let $n\in\N$, $a:\R\to\C$ be a $\frac{2\pi}{2^n}$-periodic function such that $a\in L^2\bigl[0,\frac{2\pi}{2^n}\bigr]$. For $0\leq j\leq n-1$, we define $a^{[j]}:\R\to \C$ by the formula
\begin{equation}\label{ann}
a^{[j]}(x)=a(x)r_j\Bigl(\frac{x}{2\pi}\Bigr).
\end{equation}
Then each $a^{[j]}$ is $2\pi$-periodic and belongs to $L^2[0,2\pi]$. Furthermore the Fourier coefficients of $a^{[j]}$ satisfy the following conditions$:$
\begin{align}
&\widetilde{a^{[0]}}(k)=0\ \ \text{if $2^n$ does not divide $k;$}\label{k0}
\\
&\widetilde{a^{[j]}}(k)=0\ \ \text{if $1\leq j\leq n-1$ and either $2^j$ does not divide $k$ or $2^{j+1}$ divides $k$.}\label{kj}
\end{align}
\end{lemma}

\begin{proof} Since $a\in L^2[0,2\pi]$ and $|r_j|\leq 1$, each $a^{[j]}$ belongs to $L^2[0,2\pi]$. Since $a$ is $2\pi$-periodic and each $r_j$ is $1$-periodic, each $a^{[j]}$ is $2\pi$-periodic. Next, since $a^{[0]}=a$ is $\frac{2\pi}{2^n}$-periodic, (\ref{k0}) immediately follows. From now on we assume that $1\leq j\leq n-1$. Using the definition of the Rademacher function $r_j$ and the $\frac{2\pi}{2^n}$-periodicity of $a$, it is easy to show that $a^{[j]}$ is $\frac{2\pi}{2^j}$-periodic and is orthogonal to the subspace of $L^2[0,2\pi]$ of all $\frac{2\pi}{2^{j+1}}$-periodic functions. This observation immediately implies (\ref{kj}).
\end{proof}

{\bf We are ready to prove Theorem~\ref{th2}.} \ Let $n\geq 2$ be an integer. Pick $n$ pairwise disjoint closed subintervals $I_1,\dots,I_n$ of $\bigl(0,\frac{2\pi}{2^n}\bigr)$. Now we can choose $\frac{2\pi}{2^n}$-periodic non-zero infinitely differentiable functions $a_1,\dots,a_n$ on $\R$ such that for each $j$, the support of $a_j$ restricted to $\bigl(0,\frac{2\pi}{2^n}\bigr)$ is contained in $I_j$.

Now for $1\leq j\leq n$, define $f_j=a_j^{[j-1]}$ following the notation introduced in (\ref{ann}). Then each $f_j$ is a non-zero $2\pi$-periodic infinitely differentiable function on $\R$ (the discontinuities of $r_j\bigl(\frac{2\pi}{x}\bigr)$ do not matter since they happen outside the supports of $a_k$). The restrictions we have imposed on the supports of $a_j$ imply that $f_jf_k=0$ whenever $j\neq k$. Finally, Lemma~\ref{rade} implies that $\widetilde {f_j}(m)\widetilde {f_k}(m)=0$ for every $m\in\Z$ whenever $j\neq k$. Thus $\{f_1,\dots,f_n\}$ is an $n$-element UZD-subset of $L^2[0,2\pi]$ consisting of $2\pi$-periodic infinitely differentiable functions. The proof of Theorem~\ref{th2} is complete.

\section{Pauli pairs of step functions}

We start with a number of general observations. We assume that for $b\in\C^n$ and $h_b$ is the  corresponding step function. That is, $h_b$ vanishes outside $[0,n)$, while for $1\leq j\leq n$, $h_b$  restricted to $[j-1,j)$ is the constant function $b_j$. In other words $h_b(x)=\sum\limits_{j=0}^{n-1}b_{j+1}\chi(x-j)$, where $\chi$ is the indicator function of the interval $[0,1)$. Then
$$
|\widehat{h}_b(y)|=|\widehat\chi(y)|\Bigl|\sum_{j=1}^{n}b_{j}e^{ijy}\Bigr|\ \ \text{for all}\ \ y\in\R.
$$
Hence
$$
|\widehat{h}_b(y)|^2=|\widehat\chi(y)|^2\sum_{j,k=1}^{n}b_{j}\overline{b_{k}}e^{i(j-k)y}=
|\widehat\chi(y)|^2\sum_{k=1-n}^{n-1}\rho_k(b)e^{iky}\ \ \text{for all}\ \ y\in\R,
$$
where
$$
\rho_k(b)=\sum_{j=1}^{n-k}b_{k+j}\overline{b_j}\ \ \text{if $0\leq k\leq n-1$ and}\ \
\rho_k(b)=\overline{\rho_{-k}(b)}\ \ \text{if $1-n\leq k\leq-1$}.
$$
Since two trigonometric polynomials coincide as functions if and only if their coefficients are the same, from the last two displays it follows that $|\widehat{h}_b|=|\widehat{h}_c|$ if and only if $\rho_k(b)=\rho_k(c)$ for $0\leq k\leq n$. Obviously $|h_b|=|h_c|$ if and only if $|b_j|=|c_j|$ for $1\leq j\leq n$. Finally, if $|b_j|=|c_j|$ for $1\leq j\leq n$, then automatically $\rho_0(b)=\rho_0(c)$. This is a straightforward consequence of the Parseval identity. These observations are summarized in the following lemma.

\begin{lemma}\label{stepa} The equalities $|h_b|=|h_c|$ and $|\widehat{h}_b|=|\widehat{h}_c|$ hold precisely when
\begin{align}
&b_j\overline{b_j}=c_j\overline{c_j}\ \ \text{for $1\leq j\leq n$ and}
\label{haha1}
\\
&
\sum_{j=1}^{n-k}b_{k+j}\overline{b_j}=\sum_{j=1}^{n-k}c_{k+j}\overline{c_j}\ \ \text{for $1\leq k\leq n-1$}.
\label{haha2}
\end{align}
\end{lemma}

Thus the task of finding Pauli pairs of step functions boils down to solving a system, of homogeneous degree 2 algebraic equations. Using the above lemma one can verify that for the number of steps $n\leq 3$, every Pauli pair of step functions is actually provided by Proposition~MP. We shall concentrate on the case $n=4$: the first one when the phenomenon we are looking for is possible. Even in the relatively mild specific case $n=4$, the number of variables in the system (\ref{haha1}) and (\ref{haha2}) is a bit overwhelming ($14$ real quadratic equations on $16$ real variables). Still it is possible to completely characterize the Pauli pairs of step functions with 4 steps. For the sake of clarity, a preliminary reduction is in order. Since we are interested in the genuinely 4-step situation, we may assume that $b_1\neq 0$ and $b_4\neq 0$. Since multiplying both $h_b$ and $h_c$ by non-zero complex numbers with the same absolute value does not perturb the equations $|h_b|=|h_c|$ and $|\widehat{h}_b|=|\widehat{h}_c|$, it is enough to consider the case $b_1=c_1=1$ (the general solution is easily obtained from this particular case).

\begin{proposition}\label{4steps} Let $b,c\in\C^4$ be such that $b_1=c_1=1$ and $b_4c_4\neq0$. Then the complete list of solutions of the system $|h_b|=|h_c|$ and $|\widehat{h}_b|=|\widehat{h}_c|$ is described as follows$:$
\begin{itemize}
\item[{\rm(1)}] The trivial solution$:$
\begin{equation}\label{sol1}
\text{$b=c=(1,x,y,z)$ with $x,y,z\in\C$, $z\neq0$.}
\end{equation}
This is a six $($real$)$ parametric family of solutions.
\item[{\rm(2)}] The $4$-parametric family of solutions$:$
\begin{equation}\label{sol2}
\text{$b=(1,pe^{i\phi},pe^{i\psi},e^{i\theta})$, $c=(1,pe^{i(\theta-\phi)},pe^{i(\theta-\psi)},e^{i\theta})$, where $p,\phi,\psi,\theta\in\R.$}
\end{equation}
\item[{\rm(3)}] And another $4$-parametric family of solutions$:$
\begin{align}
&\textstyle \!\!\!\!\!\!\!\!\!\!\!\!\!\!\!\!b{=}\Bigl(1,\!\frac{(r^2{-}1)\sin\psi}{r\sin(\psi{-}\phi)}e^{i(\theta{+}\phi)},\!
\frac{(r^2{-}1)\sin\phi}{\sin(\psi{-}\phi)}e^{i(2\theta{+}\psi)},re^{3i\theta}\Bigr),\
c{=}\Bigl(1,\!\frac{(r^2{-}1)\sin\psi}{r\sin(\psi{-}\phi)}e^{i(\theta{-}\phi)},\!
\frac{(r^2{-}1)\sin\phi}{\sin(\psi{-}\phi)}e^{i(2\theta{-}\psi)},re^{3i\theta}\Bigr),\label{sol3}
\\
&\qquad\text{where $r,\phi,\psi,\theta\in\R$ and $r\cdot\sin\phi\cdot\sin\psi\cdot\sin(\psi-\phi)\neq 0$}.\notag
\end{align}
\end{itemize}
\end{proposition}

The proof of the above proposition is elementary but tedious and for that reason it is banned to the Appendix. Note that we did not make the three families of solutions disjoint. Namely, some of the solutions in (\ref{sol2}) as well as in (\ref{sol3}) are actually trivial (=feature in (\ref{sol1})). Example~\ref{ex1} is obtained by plugging $r=2$, $\phi=\theta=\frac{\pi}3$ and $\psi=\frac{2\pi}3$ into (\ref{sol3}). Note that the MP$^1$-pairs in the above proposition are all collected in (\ref{sol2}), while the MP$^2$ pairs are given by (\ref{sol3}) with the additional constraint $e^{2i\theta}=1$.

\section{Concluding remarks and open problems}

The first problem is obvious.

\begin{question}\label{qq1}
Describe Pauli pairs of step functions with arbitrary number of steps.
\end{question}

Let us call a set $P\subset L^2(\R)$ a {\it Pauli set} if the cardinality of $P$ is at least 2 and every distinct
elements $f$ and $g$ of $P$ form a Pauli pair. We have already observed that Theorem~\ref{th1} provides a huge Pauli set. The author has a proof of the fact that a Pauli set of step functions with $n$ steps has at most $2^n$ elements. Furthermore any Pauli set of functions with bounded supports is totally disconnected. The proofs are not included since I strongly believe both statements are miles away from optimal. For instance, I think that a Pauli set of functions with bounded supports must be discrete (if not finite).

\begin{question}\label{qq2}
Is it true that every Pauli set of functions with bounded supports is discrete $($and hence countable$)$?
\end{question}

\begin{question}\label{qq3}
What exactly is the biggest cardinality of a Pauli set of step functions with $n$ steps?
\end{question}

\bigskip

{\bf Acknowledgements.} \ The author is grateful to the referee for useful comments and suggestions, which helped to improve the presentation.

\small

\section{Appendix: Proof of Proposition~\ref{4steps}}

By Lemma~\ref{stepa}, the equalities $|h_b|=|h_c|$ and $|\widehat{h}_b|=|\widehat{h}_c|$ hold if and only if
(\ref{haha1}) and (\ref{haha2}) are satisfied. In our specific case $n=4$ and $b_1=c_1=1$ the system (\ref{haha1}) and (\ref{haha2}) can be rewritten as:
\begin{equation}\label{haha3}
\text{$|b_2|=|c_2|$, $|b_3|=|c_3|$, $b_4=c_4$, $b_4\overline{b_2}+b_3=c_4\overline{c_2}+c_3$,
$b_4\overline{b_3}+b_3\overline{b_2}+b_2=c_4\overline{c_3}+c_3\overline{c_2}+c_2$.}
\end{equation}
In view of $|b_2|=|c_2|$, $|b_3|=|c_3|$ and $b_4=c_4$, we can write $b_2=ps$, $c_2=pt$, $b_3=qu$, $c_3=qv$ and $b_4=c_4=rx$, where $p,q,r\in\R$, $s,t,u,v,x\in\T$ and $r\neq 0$. In this new notation the first three equations in (\ref{haha3}) are satisfied automatically and (\ref{haha3}) becomes equivalent to
\begin{align}
&\textstyle rp\bigl(\frac{x}{s}-\frac{x}{t}\bigr)+q(u-v)=0,\label{re1}
\\
&\textstyle rq\bigl(\frac{x}{u}-\frac{x}{v}\bigr)+pq\bigl(\frac{u}{s}-\frac{v}{t}\bigr)+p(s-t)=0.\label{re2}
\end{align}

We start by getting rid of an easy degenerate case.

{\bf Case 1:} $pq(s-t)(u-v)=0$. The system (\ref{re1}) and (\ref{re2}) ensures that in this case at least one of the following statements is satisfied:
\begin{itemize}
\item $p=q=0$;
\item $p=0$ and $u=v$;
\item $q=0$ and $s=t$;
\item $s=t$ and $u=v$.
\end{itemize}
In each of these four cases, we have $b=c$. Thus in Case~1 we only have (a particular case of) the trivial solution (\ref{sol1}).

The further analysis of (\ref{re1}) and (\ref{re2}) relies upon the following elementary fact, the proof of which is left as an exercise to the reader.

\begin{lemma}\label{cheva}
For $\alpha,\beta,\gamma,\delta\in\T$, the complex numbers $0$, $\alpha-\beta$ and $\frac1\gamma-\frac1\delta$ lie on one line if and only if either $\alpha=\beta$ or $\gamma=\delta$ or $\alpha\beta\gamma\delta=1$.
\end{lemma}

{\bf Case 2:} $pq(s-t)(u-v)\neq0$. In this case (\ref{re1}) guarantees that $0$, $u-v$ and $\frac{x}{s}-\frac{x}{t}$ lie on the same line. Since $u\neq v$ and $s\neq t$, Lemma~\ref{cheva} implies that $uv\frac{s}{x}\frac{t}{x}=1$. That is, 
\begin{equation}
\label{xsq}
stuv=x^2.
\end{equation}
By the same Lemma~\ref{cheva}, (\ref{xsq}) guarantees that $0$, $s-t$ and $\frac{x}{u}-\frac{x}{v}$ lie on the same line $L$ in $\C$. Thus by (\ref{re2}), $\frac{u}{s}-\frac{v}{t}\in L$. Hence $0$, $s-t$ and $\frac{u}{s}-\frac{v}{t}$ lie on the same line. Since $s\neq t$, Lemma~\ref{cheva} implies that either $\frac{u}{s}=\frac{v}{t}$ or $st\frac{s}{u}\frac{t}{v}=1$. Thus at least one of the following two conditions must be satisfied:
$$
\text{$ut=vs$\ \ or\ \ $uv=(st)^2$.}
$$
Thus Case~2 splits into two subcases.

{\bf Case 2a:} $stuv=x^2$ and $ut=vs$. Hence $ut=vs=\pm x$. Since replacing $(r,x)$ by $(-r,-x)$ does not actually change $b$  and $c$, we can assume that $r>0$. Easy cancellation shows that (\ref{re1}) and (\ref{re2}) can be rewritten as $q=rp$ and $p=rq$ if $ut=vs=x$ and as $q=-rp$ and $p=-rq$ if $ut=vs=-x$.
In the first case we have $r=1$ and $p=q$, while in the second case we have $r=1$ and $q=-p$.
Since $s,u,t\in\T$, we can write $s=e^{i\phi}$, $u=e^{i\psi}$ and $x=e^{i\theta}$ for $\phi,\psi,\theta\in\R$. Then  $v=\frac{x}{s}=e^{i(\theta-\phi)}$ and $t=\frac{x}{u}=e^{i(\theta-\psi)}$. If $ut=vs=x$, we have $r=1$ and $q=p$ and therefore $b=(1,pe^{i\phi},pe^{i\psi},e^{i\theta})$ and $c=(1,pe^{i(\theta-\phi)},pe^{i(\theta-\psi)},e^{i\theta})$, which is exactly the family (\ref{sol2}) of solutions. If $ut=vs=-x$, we have $r=1$ and $q=-p$ and therefore $b=(1,pe^{i\phi},-pe^{i\psi},e^{i\theta})$ and $c=(1,pe^{i(\theta-\phi)},-pe^{i(\theta-\psi)},e^{i\theta})$. The change of parametrization $(p,\phi,\psi,\theta)\mapsto (p,\phi,\pi+\psi,\theta)$ shows that this family of solutions is the same old (\ref{sol2}). It remains to consider the following case.

{\bf Case 2b:} $stuv=x^2$, $sv\neq ut$ and $uv=(st)^2$. It follows that $(st)^3=x^2$. Hence there exists $w\in\T$ such that $x=w^3$ and $st=w^2$. Denote $\alpha=\frac{s}{w}$ and $\beta=\frac{u}{w^2}$. In this notation $s=\alpha w$, $t=\frac{w}{\alpha}$, $u=\beta w^2$ and $v=\frac{w^2}{\beta}$. Substituting these expressions into the system (\ref{re1}) and (\ref{re2}), we can rewrite it in the following equivalent way:
\begin{align}
&rp(\alpha-\overline{\alpha})+q(\beta-\overline{\beta})=0,\label{re3}
\\ 
& \textstyle p(\alpha-\overline{\alpha})+pq\bigl(\frac{\beta}{\alpha}-\frac{\overline{\beta}}{\overline{\alpha}}\bigr)-
rq(\beta-\overline{\beta})=0.\label{re4}
\end{align}
Since $w,\alpha,\beta\in\T$, we can write $\alpha=e^{i\phi}$, $\beta=e^{i\psi}$ and $w=e^{i\theta}$ for some $\phi,\psi,\theta\in\R$. The relations $s\neq t$, $u\neq v$ and $sv\neq ut$ are equivalent to $\sin\phi\neq0$, $\sin\psi\neq 0$ and $\sin(\psi-\phi)\neq 0$ respectively. The equations (\ref{re3}) and (\ref{re4}) now read:
\begin{align}
&rp\sin\phi+q\sin\psi=0,\label{re5}
\\
&p\sin\phi+pq\sin(\psi-\phi)-rq\sin\psi=0.\label{re6}
\end{align}
Using the relations $\sin\phi\neq0$, $\sin\psi\neq 0$ and $\sin(\psi-\phi)\neq 0$, we can easily solve this system. Namely, (\ref{re5}) and (\ref{re6}) are equivalent to
$$
p=\frac{(r^2-1)\sin\psi}{r\sin(\psi-\phi)}\ \ \text{and}\ \ q=\frac{(r^2-1)\sin\phi}{\sin(\psi-\phi)}.
$$
Recovering $b$ and $c$ (for instance $b_4=c_4=rx=rw^3=re^{3i\theta}$, etc.), we arrive to the family (\ref{sol3}) of solutions.

The proof of Proposition~\ref{4steps} is complete.

\vfill\break

\small\rm

\vskip1truecm

\scshape

\noindent Stanislav Shkarin

\noindent Queens's University Belfast

\noindent Pure Mathematics Research Centre

\noindent University road, Belfast, BT7 1NN, UK

\noindent E-mail address: \qquad {\tt s.shkarin@qub.ac.uk}

\end{document}